\documentclass[letterpaper, 10 pt, conference]{ieeeconf}  

\IEEEoverridecommandlockouts                              

\usepackage{amsfonts}
\usepackage{amsmath,amssymb,amsthm}
\usepackage{cite}
\usepackage{graphicx}
\usepackage{subfig}
\usepackage{mwe}
\usepackage{scalefnt}
\usepackage[font=small]{caption}
\allowdisplaybreaks
\newtheorem{theorem}{Theorem}

\title{\LARGE \bf
Compositionality of Linearly Solvable Optimal Control in Networked Multi-Agent Systems 
}

 \author{Lin Song$^{1}$, Neng Wan$^{1}$, Aditya Gahlawat$^{1}$, Naira Hovakimyan$^{1}$, and Evangelos A. Theodorou$^{2}$
   \thanks{$^{1}$Lin Song, Neng Wan, Aditya Gahlawat and Naira Hovakimyan are with the Department of Mechanical Science and Engineering, University of Illinois at Urbana-Champaign, Urbana, IL 61801 USA {\tt\footnotesize \{linsong2, nengwan2, gahlawat, nhovakim\}@illinois.edu}}%
   \thanks{$^{2}$Evangelos A. Theodorou is with the Department of Aerospace Engineering, Georgia Institute of Technology, Atlanta, GA 30332 USA {\tt\footnotesize \{evangelos.theodorou\}@gatech.edu}}
 }

\begin{document}

\maketitle
\thispagestyle{empty}
\pagestyle{empty}

\begin{abstract}

In this paper, we discuss the methodology of generalizing the optimal control law from learned component tasks to unlearned composite tasks on  Multi-Agent Systems (MASs), by using the linearity composition principle of linearly solvable optimal control (LSOC) problems. The proposed approach achieves both the compositionality and optimality of control actions simultaneously within the cooperative MAS framework in both discrete and continuous-time in a sample-efficient manner, which reduces the burden of 
re-computation of the optimal control solutions for the new task on the MASs. We investigate the application of the proposed approach on the MAS with coordination between agents. The experiments show feasible results in investigated scenarios, including both discrete and continuous dynamical systems for task generalization without resampling.

\end{abstract}

\section{INTRODUCTION}

Stochastic optimal control problems \cite{kumar2015stochastic}, located at the intersection of Reinforcement Learning (RL) \cite{sutton2018reinforcement} and control theory, attracted a wide community of researchers over the last few years. However, they are difficult and computationally expensive to solve for high dimensional problems \cite{blondel2000survey}. To overcome the challenges in computational efficiency, some approximation-based approaches were introduced to achieve optimal control solutions, such as cost parameterization \cite{bertsekas1995neuro}, path integral formulation \cite{theodorou2010generalized, theodorou2010reinforcement}, value function approximation \cite{powell2011review} and policy approximation \cite{sutton2000policy}. In 
\cite{todorov2009efficient}, an exponential transformation was applied to represent the cost function in a form of Kullback–Leibler (KL) divergence between probability distributions, such that the optimal solution was obtained in a linear form. Furthermore, the linearly solvable optimal control
(LSOC) problem was generally formulated in \cite{dvijotham2013linearly} to summarize the class of optimal control problems whose solutions can be obtained by solving a reduced linear equation. The LSOC problems allow for control composition and path-integral representation of the optimal solution \cite{dvijotham2011unified}. The compositionality of LSOC can also improve the computational efficiency by generalizing existing controllers towards constructing new controllers for solving a general class of problems according to the contribution of each component problem  \cite{todorov2009compositionality,pan2015sample}. The composite controller can be obtained by immediate computation 
via weighing on existing controllers, and generalizing to complex problems in a certain class without resampling. However, the aforementioned work on the compositionality of LSOC only considers single-agent problems, and the compositionality of the optimal control actions in networked MASs is seldomly discussed in prior work.

Networked multi-agent systems allow team coordination and provide flexibility in many application scenarios, including robotics \cite{8960572}, sensor networks \cite{zheng2018average} and transportation systems \cite{liu2017intelligent}, etc. Theoretical foundations of MASs research include the framework of Decentralized Partially Observable Markov Decision Process (Dec-POMDP) in  \cite{amato2013decentralized} for finding an optimal policy for MASs, and Distributed Constraint Optimization Problems (DCOP) in \cite{maheswaran2004taking} for handling the coordination of multiple agents. Though broadly applied, MAS 
problems are typically intractable to solve and scale. In \cite{chades2002heuristic}, the author proposed factored representations as extension of Dec-POMDP for scalability. 
The Decentralized Partially Observable Semi-Markov Decision Process (Dec-POSMDP) framework  was proposed in \cite{omidshafiei2017decentralized} for efficient computation in large scale discrete- and continuous-time problems as an extended version of Dec-POMDP. Meanwhile, some RL techniques were integrated in the MAS framework for achieving coordination and optimality simultaneously \cite{busoniu2008comprehensive} from the perspective of machine learning
\cite{stone2000multiagent,zhang2018fully,zhang2019multi} and game theory \cite{lanctot2017unified,bredin2000game}. 

This paper presents integration of compositionality into the optimal control framework of networked MASs; numerical examples are provided validating the proposed approach in both discrete and continuous time.
This paper is organized as follows: Section \ref{2_sec} formulates the control problem in networked MAS setting; Section \ref{3_sec} investigates the main results on achieving both compositionality and optimality in MASs and the task generalization capacity; Section \ref{4_sec} provides numerical examples validating the proposed approach; the conclusion and future work are discussed in Section \ref{5_sec}.

\section{PROBLEM FORMULATION}\label{2_sec}
\subsection{Stochastic Optimal Control of Single-Agent Systems}
\subsubsection{Discrete-time Systems}
For the discrete-time single-agent problems, also known as Markov Decision Processes (MDPs), the passive dynamics are defined by the transition probability, i.e. $x'_i \sim p_i(\cdot | x_i)$
 and the controlled dynamics $x'_i \sim u_i(\cdot|x_i) = p_i(\cdot|x_i,u_i)$ for agent $i$.

The running cost for the agent $i$ includes state cost $q_i(\cdot) \ge 0$ and action cost, defined via the KL divergence between the controlled and passive dynamics. This way, the running cost is in the form of 
$
    c_i(x_i,u_i) = q_i(x_i) + \textrm{KL}(u(\cdot|x_i)\|p_i(\cdot|x_i)) 
    = q_i(x_i)+\mathbb{E}_{x'_i \sim u(\cdot|x_i)}\Big[\textrm{log} \frac{u(x'_i|x_i)}{p(x'_i|x_i)}\Big].
$

Let $t_f$ denote the terminal time, and $\mathcal{I}_i,\mathcal{B}_i$ represent sets of interior and boundary states, respectively (i.e. $x_{i} \in \mathcal{I}_i, x_i^{t_f} \in \mathcal{B}_i$). The cost-to-go function for first-exit total cost setting introduced in \cite{todorov2009efficient} calculates the expected cumulative cost starting from state $x_i^{t_0}$ and time $t_0$ under control law $u_i$ in the form of $
    J_i^{u_i}(x_i^{t_0},t_0) = \mathbb{E}^{u_i}\big[\phi_i(x_i^{t_f})+\sum\nolimits_{\tau = t_0}^{t_f-1}c_i(x_i^{\tau},u_i^{\tau})\big],
$
with $\phi_i(x_i)$ denoting the final cost function. In this formulation, the terminal time $t_f$ is determined online, when the target state $x_i^{t_f} \in \mathcal{B}_i$ is first reached. 

For discrete-time stochastic optimal control problem, the goal is to obtain the optimal policy $u_i^*$ such that the dynamical system starting from the state $x_i^{t_0}$ and time $t_0$ behaves optimally thereafter, and thus the value function defined as the optimal cost-to-go function by $V_i(x_i) = \min{u_i}J_i^{u_i}(x_i^{t_0},t_0)$ satisfies the Bellman equation:
\begin{align}\label{eq:bellman_discrete_single}
    V_i(x_i) &= \min_{u_i}\{c_i(x_i,u_i)+\mathbb{E}_{x'_i\sim u_i(\cdot|x_i)}[V_i(x'_i)]\}.
\end{align}  

With the exponential transformation provided by the desirability function $Z_i(x_i) = \textrm{exp}(-V_i(x_i))$ for agent $i$, the optimization problem in \eqref{eq:bellman_discrete_single} has a linear solution in the form of \begin{align}\label{eq:dis_single_opt}
    u_i^*(x'_i|x_i) &= \frac{p(x'_i|x_i)Z_i(x'_i)}{\mathbf{G}[Z_i](x_i)} =  \frac{p(x'_i|x_i)Z_i(x'_i)}{\sum_{x'_i}p(x'_i|x_i)Z_i(x'_i)}.
\end{align}   

The Bellman equation in \eqref{eq:bellman_discrete_single} is thus reduced to a linear equation of the form \begin{equation}\label{eq:bellman_dis}
    \textrm{exp}(c_i(x_i))Z_i(x_i) = \mathbf{G}[Z_i](x_i),
\end{equation}
under the optimal control law $u_i^*$.


\subsubsection{Continuous-time Systems}
For the continuous-time dynamical system of agent $i$, the dynamics are usually described as an It$\hat{\textrm{o}}$ diffusion process: \begin{align}
    dx_i &= g_i(x_i,t)dt+B_i(x_i)[u_i(x_i,t)dt+\sigma_i d\omega_i] \label{eq:cont_dyn}\\
    &= f_i(x_i,u_i,t)dt+F_i(x_i,u_i)d\omega_i.
\nonumber
\end{align}

Similar to the discrete-time scenario, the immediate cost penalizes in both state and action, in the form of $
    c_i(x_i,u_i) = q_i(x_i) + \frac{1}{2}u_i(x_i,t)^\top R_i u_i(x_i,t),
$
with $R_i$ being a positive definite matrix. The cost-to-go function calculates the expected cumulative cost starting from state $x_i$ under control $u_i$. Through defining the value function $V_i$ as the optimal cost-to-go function and a stochastic second-order differentiator $\mathcal{L}_{(u_i)}[V_i] = f_i^\top \nabla_{x_i}V_i + \frac{1}{2}\textrm{trace}(F_i F_i^\top \nabla^2_{x_ix_i}V_i)$, the stochastic Hamilton-Jacobian-Bellman (HJB) equation takes the form of
$
    0 = \min_{u_i}\{c_i(x_i,u_i) + \mathcal{L}_{(u_i)}[V_i](x_i)\}.
$

Defining the desirability function  $Z_i(x_i,t) = \textrm{exp}(-V_i(x_i,t))$, and under the nonlinearity cancellation conditions, the optimal control law in continuous-time is reduced to \begin{equation}\label{eq:cont_single_opt}
    u_i^*(x_i,t) = \sigma_i \sigma_i^\top  B_i(x_i)^\top \frac{\nabla_{x_i}Z_{i}(x_i,t)}{Z_i(x_i,t)},
\end{equation} 
satisfying the transformed HJB equation in the form of \begin{equation}\label{eq:hjb_cont}
    0 = \mathcal{L}[Z_i]-q_iZ_i,
\end{equation}
where $\sigma_i$ and $B_i$ depict the continuous dynamics as in \eqref{eq:cont_dyn}.

\subsubsection{Composition Theory for Linearly Solvable Optimal Control} 
Considering the optimal control law takes the form of \eqref{eq:dis_single_opt} and \eqref{eq:cont_single_opt} and based on the linearity of equations, we can develop the compositionality of optimal control laws.

Assume there are $F$ component problems governed by identical dynamics, cost rates,   set of interior and boundary states, and differed at the final cost. For agent $i$, let $Z_i^{\{f\}}(x_i)$ denote the desirability function in the component problem $f$, $u_{i}^*{^{\{f\}}}$ denote the corresponding optimal control law and $\omega_i^{\{f\}}$ denote the weight of the problem $f$. The desirability function satisfies the following weighted combination relationship between the component problems and the composite problem on the boundary states:\begin{equation}
    \label{eq:des_relation}
    Z_{i}(x_i) =  \sum\nolimits_{f = 1}^{F} \omega_i^{\{f\}}Z_{i}^{\{f\}}(x_i),
\end{equation}
which is equivalent to the following relationship on the final cost $h$: \begin{equation}\label{eq:final_cost_relation}
    h_{i}(x_i) = -\log( \sum\nolimits_{f=1}^{F}\omega_i^{\{f\}} \textrm{exp}(-h_i^{\{f\}}(x_i))).
\end{equation}

Since the desirability function $Z_i$ solves a linear equation by nature (\eqref{eq:bellman_dis} in discrete-time and \eqref{eq:hjb_cont} in continuous-time), once condition \eqref{eq:des_relation} holds on the boundary, it holds everywhere, and the compositionality of the optimal control laws follows in a task generalization setting. This formulation creates a straightforward application scenario in the physical world. For example, a UAV may deliver packages to different warehouses with different terminal costs. 

In discrete time, the compositionality of desirability function directly applies to the control law in equation \eqref{eq:dis_single_opt}, and we have $
    u_i^*(\cdot|x_i) =  \sum\nolimits_{f=1}^{F}\frac{\omega_f Z_i^{\{f\}}(x_i)}{\sum\nolimits_{e=1}^{F}\omega_e Z_i^{\{e\}}(x_i)} u_i^{*^{\{f\}}}(\cdot|x_i),
$
where $u_i^*(x'_i|x_i)$ denotes the transition probability from $x_i$ to $x'_i$ under control.

In continuous time, the desirability functions between component and composite problems satisfy the following relationship for all the states:
\begin{equation}\label{eq:des_relation_cont}
    Z_{i}(x_i, t) =  \sum\nolimits_{f = 1}^{F} \omega_i^{\{f\}}Z_{i}^{\{f\}}(x_i,t),
\end{equation}
and the corresponding compositionality on control actions becomes
\begin{equation}
    u_i^*(x_i,t) = \sum\nolimits_{f=1}^{F}\frac{\omega_i^{\{f\}}Z_i^{\{f\}}(x_i,t)}{\sum\nolimits_{e=1}^{F} \omega_i^{\{e\}}Z_i^{\{e\}}(x_i,t)} u_i^{*^{\{f\}}}(x_i,t).
\end{equation}

\subsection{Stochastic Optimal Control of MASs}
\subsubsection{MASs and Factorial Subsystems}
We use a connected and undirected graph $\mathcal{G} = \{\mathcal{V,E}\}$ to describe the networked MAS, where $v_i \in \mathcal{V}$ denotes the graph vertex representing agent $i$, and $e_{ij} \in \mathcal{E}$ represents the edge connecting agents $i$ and $j$, enabling bidirectional information exchange. Thus, the networked MAS is factorized into multiple subsystems (i.e. $\bar{\mathcal{N}}_j$), each of which includes a central agent $\{j\}$ and index sets of its neighboring agents $\mathcal{N}_
j$. The notation considers both existing coordination between agents and the complexity in computation of optimal control laws. An illustrative example of a MAS and factorial subsystems is provided in Figure \ref{fig_fact}. Here, $\bar{x}_i$ denotes the joint states for the factorial subsystem $i$, $\overline{\mathcal{I}}_i,\overline{\mathcal{B}}_i$ represent the sets of interior and boundary states for the factorial subsystem $i$, respectively (i.e. $\bar{x}_{i} \in \overline{\mathcal{I}_i}, \bar{x}_{i}^{t_f} \in \overline{\mathcal{B}}_i$). 
\begin{figure}[h!]
\centering
\vspace{-1mm}
\includegraphics[width=2.0in]{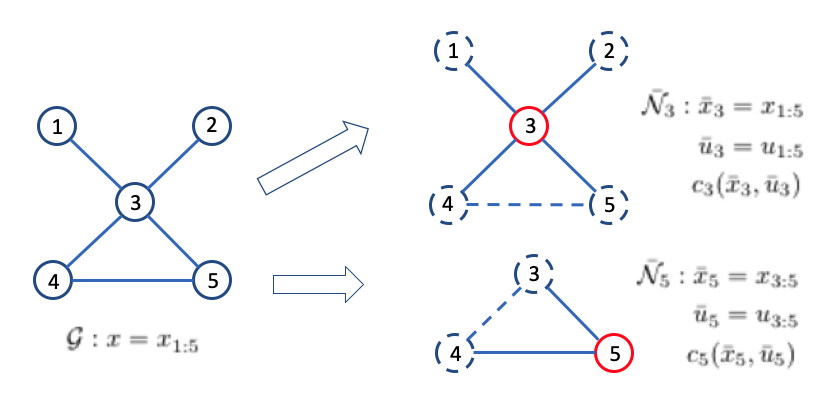}
\caption{MAS $\mathcal{G}$ and factorial subsystems $\bar{\mathcal{N}}_3$ and $\bar{\mathcal{N}}_5$.}
\label{fig_fact}
\end{figure}

For the cooperative control framework, we calculate the local control action $u_j$ depending on the local observation of agent $j$ within the factorial subsystem $\bar{\mathcal{N}}_j$, where $j$ is the central agent. In this representation, dependency on irrelevant states is removed and exponential computational complexity on global states is avoided. More detailed interpretations on our distributed framework can be found in \cite{wan2021distributed}.

In discrete time, the uncontrolled dynamics for the factorial subsystem $i$ can be interpreted as:
\begin{equation}\label{joint_dynamics_discrete}
  \bar{x}'_i \sim \bar{p}_i(\cdot|\bar{x}_i) = \prod\nolimits_{j \in \bar{\mathcal{N}}_i} p_j(\cdot|x_j),
\end{equation}
where $\bar{p}_i(\bar{x}'_i|\bar{x}_i)$ denotes the transition probability from factorial states $\bar{x}_i$ to $\bar{x}'_i$. The running cost follows a similar definition as in the single-agent problem in the form of
\begin{align}
    c_i(\bar{x}_i,\bar{u}_i) &= q_i(\bar{x}_i) + \textrm{KL}(\bar{u}_i(\cdot|\bar{x}_i) || \bar{p}_i(\cdot|\bar{x}_i)) \nonumber\\
    &= q_i(\bar{x}_i) +\sum\nolimits_{j \in \bar{\mathcal{N}}_i} \textrm{KL}(u_j(\cdot|\bar{x}_i) || p_j(\cdot|x_j)) \label{joint_cost_rate_discrete}.
\end{align}

In continuous time, the uncontrolled dynamics for the factorial subsystem $i$ can be represented by the diffusion process:
\begin{equation}\label{joint_dynamics_continuous}
    d\bar{x}_i(\tau) = \bar{f}_i(\bar{x}_i,\tau)d\tau + \bar{B}_i(\bar{x}_i)\bar{\sigma}_i\cdot d\bar{\omega}_i(\tau).
\end{equation}
The running cost is defined correspondingly in the form of
\begin{equation}\label{joint_cost_continuous}
    c_i(\bar{x}_i,\bar{u}_i) = q_i(\bar{x}_i) + \frac{1}{2}\bar{u}_i(\bar{x}_i,t)^\top \bar{R}_i \bar{u}_i(\bar{x}_i,t).
\end{equation}

\subsubsection{Joint Optimal Control Actions in Factorial Subsystems}\label{sec_to_ref}
In discrete time, let  $Z_i(\bar{x}_i) = \textrm{exp}(-V_i(\bar{x}_i))$ define the desirability function for the factorial subsystem $i$. Then we have  the joint optimal control action \begin{equation}\label{eq:joint_action_discrete}
    \bar{u}_i^{*}({\bar{x}'_i}|\bar{x}_i) = \frac{\bar{p}_i(\bar{x}'_i|\bar{x}_i)Z_{i}(\bar{x}'_i)}{\sum\nolimits_{\bar{x}'_i} \bar{p}_i(\bar{x}'_i|\bar{x}_i)Z_{i}(\bar{x}'_i)},
\end{equation}
where $\bar{u}_i({\bar{x}'_i}|\bar{x}_i)$ denotes the transition probability from factorial subsystem states $\bar{x}_i$ to $\bar{x}'_i$ in the controlled dynamics.
According to the factorization architecture, the local control action can be obtained immediately by marginalization, and we have $
    u_i^*(x'_i|\bar{x}_i) = \sum\nolimits_{j\in \mathcal{N}_i} \bar{u}_i^{*}(x'_i,x'_{j\in \mathcal{N}_i}|\bar{x}_i).
$

In continuous time, 
define $Z_i(\bar{x}_i,t) = \textrm{exp}[-V_i(\bar{x}_i,t)/\lambda_i]$ as the desirability function of the factorial subsystem $i$ dependent on the joint states $\bar{x}_i$, where $\lambda_i \in \mathbb{R}$ is a scalar. Then the joint optimal control action can be obtained as
\begin{equation}\label{eq:joint_action_continuous}
    \bar{u}_i^{*}(\bar{x}_i, t) = \bar{\sigma}_i \bar{\sigma}_i^\top \bar{B}_i^\top (\bar{x}_i) \frac{\nabla_{\bar{x}_i} Z_i(\bar{x}_i,t)}{Z_i(\bar{x}_i,t)}.
\end{equation}
Similarly, the local control action can be obtained by sampling on the marginal distribution of joint optimal control action in continuous-time.

\section{MAIN RESULTS}\label{3_sec}
The compositionality of optimal control laws can be extended to networked MASs in a cooperative control framework, and the main results are formulated through the following two theorems.
\begin{theorem}\label{theorem_dis}
(Discrete-time MAS compositionality) 
For $F$ multi-agent LSOC problems in discrete-time governed by the same dynamics \eqref{joint_dynamics_discrete}, the running cost rates \eqref{joint_cost_rate_discrete}, and the set of interior states $\overline{\mathcal{I}}_i$, but various terminal costs denoted by $h_i^{\{f\}}$,  let $\bar{x}_i^{d^{\{f\}}}$ denote the state targets of the component problem $f$  and $\bar{x}_i^d$ denote the state target for the new task (composite problem) for the factorial subsystem $i$.  Define the  weights \begin{equation}\label{eq:discrete_weight}
    \bar{\omega}_i^{\{f\}} = \textrm{exp}(-\frac{1}{2}(\bar{x}_i^{d}-\bar{x}_i^{d^{\{f\}}})^\top \mathbf{P}(\bar{x}_i^{d}-\bar{x}_i^{d^{\{f\}}})) ,
\end{equation}
with $\mathbf{P}$ being a positive definite diagonal matrix representing the kernel widths. The terminal cost for the new task becomes \begin{equation}\label{thm1_final_cost_composition}
    h_i(\bar{x}_{i}^{t_f}) = -\log(\sum\nolimits_{f=1}^F \tilde{\omega}_i^{\{f\}} exp(-h_i^{\{f\}}(\bar{x}_{i}^{t_f}))),
\end{equation} where $\bar{x}^{t_f}$ denotes the boundary states, and the coefficients $\tilde{\omega}_i^{\{f\}} = \frac{\bar{\omega}_i^{\{f\}}}{\sum_{f=1}^{F}\bar{\omega}_i^{\{f\}}} $ can be interpreted as probability weights. The control law solving the new problem is obtained by composition of the existing controllers
$
    \bar{u}_i^{*}({\bar{x}'_i}|\bar{x}_i) = \sum\nolimits_{f=1}^{F}\bar{W}_i^{\{f\}}(\bar{x}_i)\bar{u}_i^{*^{\{f\}}}(\bar{x}'_i|\bar{x}_i),
$
with 
$
    \bar{W}_i^{\{f\}}(\bar{x}_i) = \frac{\tilde{\omega}_i^{\{f\}}\mathcal{H}_i^{\{f\}}(\bar{x}_i)}{\sum\nolimits_{e=1}^{F}\tilde{\omega}_i^{\{e\}}\mathcal{H}_i^{\{e\}}(\bar{x}_i)}$ and $
\mathcal{H}_i^{\{f\}}(\bar{x}_i) = \sum\nolimits_{\bar{x}'_i}\bar{p}_i(\bar{x}'_i|\bar{x}_i) Z_i^{\{f\}}(\bar{x}'_i).
$
\end{theorem}
\begin{proof} 
From the composition on the final cost relation in equation \eqref{thm1_final_cost_composition}, and the exponential transformation given by the desirability function, we have \begin{equation}\label{eq:thm1_des_relation_boundary}
    Z_{i}(\bar{x}_{i}^{t_f}) = \sum\nolimits_{f = 1}^{F} \tilde{\omega}_i^{\{f\}}Z_{i}^{\{f\}}(\bar{x}_{i}^{t_f}).
\end{equation}

Furthermore, considering that the desirability function $ Z_{i}(\bar{x}_{i})$ solves a linear PDE, as long as the solution holds on the boundary states, the linear combination of the solutions holds everywhere. Thus, we have  \begin{equation}\label{eq:thm1_des_relation_interior}
    Z_{i}(\bar{x}_{i}) = \sum\nolimits_{f = 1}^{F} \tilde{\omega}_i^{\{f\}}Z_{i}^{\{f\}}(\bar{x}_{i}).
\end{equation} 

For the composite problem, the joint optimal control actions introduced in equation \eqref{eq:joint_action_discrete} can thus be reduced by the  compositionality of the desirability functions:
\begin{align}
    &\quad \bar{u}_i^{*}({\bar{x}'_i}|\bar{x}_i) = \frac{\bar{p}_i(\bar{x}'_i|\bar{x}_i)Z_{i}(\bar{x}'_i)}{\sum\nolimits_{\bar{x}'_i} \bar{p}_i(\bar{x}'_i|\bar{x}_i)Z_{i}(\bar{x}'_i)} \nonumber \\
    &= \frac{\bar{p}_i(\bar{x}'_i|\bar{x}_i)\sum\nolimits_{f=1}^{F} \tilde{\omega}_i^{\{f\}}Z_i^{\{f\}}(\bar{x}'_i)}{\sum\nolimits_{\bar{x}'_i} \bar{p}_i(\bar{x}'_i|\bar{x}_i)\sum\nolimits_{e=1}^{F} \tilde{\omega}_i^{\{e\}}Z_{i}^{\{e\}}(\bar{x}'_i)} \nonumber\\
    &= \frac{\sum\nolimits_{f=1}^{F}\tilde{\omega}_i^{\{f\}}\bar{p}_i(\bar{x}'_i|\bar{x}_i) Z_i^{\{f\}}(\bar{x}'_i)}{\sum\nolimits_{e=1}^{F} \tilde{\omega}_i^{\{e\}}\Big[\sum\nolimits_{\bar{x}'_i} \bar{p}_i(\bar{x}'_i|\bar{x}_i)Z_i^{\{e\}}(\bar{x}'_i)\Big]} \nonumber\\
    &= \sum\nolimits_{f=1}^{F}{\frac{\tilde{\omega}_i^{\{f\}}\bar{p}_i(\bar{x}'_i|\bar{x}_i) Z_i^{\{f\}}(\bar{x}'_i)}{\sum\nolimits_{e=1}^{F} \tilde{\omega}_i^{\{e\}}\Big[\sum\nolimits_{\bar{x}'_i} \bar{p}_i(\bar{x}'_i|\bar{x}_i)Z_i^{\{e\}}(\bar{x}'_i)\Big]}} \label{eq:line4}.
    \end{align}
    
    The terms in the denominator and numerator of \eqref{eq:line4} can be multiplied by a normalization term $\sum\nolimits_{\bar{x}'_i} \bar{p}_i(\bar{x}'_i|\bar{x}_i)Z_i^{\{f\}}(\bar{x}'_i)$ simultaneously, and after term reordering, the equation is reduced to
    \begin{align}
    \bar{u}_i^{*}({\bar{x}'_i}|\bar{x}_i) 
    &= \sum\nolimits_{f=1}^{F}\frac{\tilde{\omega}_i^{\{f\}}\mathcal{H}_i^{\{f\}}(\bar{x}_i)}{\sum\nolimits_{e=1}^{F}\tilde{\omega}_i^{\{e\}}\mathcal{H}_i^{\{e\}}(\bar{x}_i)}\bar{u}_i^{*^{\{f\}}}(\bar{x}'_i|\bar{x}_i)\\
    &= \sum\nolimits_{f=1}^{F}\bar{W}_i^{\{f\}}(\bar{x}_i)\bar{u}_i^{*^{\{f\}}}(\bar{x}'_i|\bar{x}_i),
\end{align}
with 
\begin{align*}
    \bar{W}_i^{\{f\}}(\bar{x}_i) &= \frac{\tilde{\omega}_i^{\{f\}}\mathcal{H}_i^{\{f\}}(\bar{x}_i)}{\sum\nolimits_{e=1}^{F}\tilde{\omega}_i^{\{e\}}\mathcal{H}_i^{\{e\}}(\bar{x}_i)},\\
\mathcal{H}_i^{\{f\}}(\bar{x}_i) &= \sum\nolimits_{\bar{x}'_i}\bar{p}_i(\bar{x}'_i|\bar{x}_i) Z_i^{\{f\}}(\bar{x}'_i).
\end{align*} 
\end{proof}
\begin{theorem}\label{theorem_cont}
(Continuous-time MAS compositionality) For $F$ multi-agent LSOC problems in continuous-time governed by the same dynamics \eqref{joint_dynamics_continuous}, running cost rates \eqref{joint_cost_continuous}, and the set of interior states $\overline{\mathcal{I}}_i$, but various terminal costs denoted by $h_i^{\{f\}}$, let $\bar{x}_i^{d^{\{f\}}}$ denote the state targets of the component problem $f$ and $\bar{x}_i^d$ denote the state target for the new task (composite problem)  for the factorial subsystem $i$.  Define the weights \begin{equation}\label{eq:continuous_weight}
    \bar{\omega}_i^{\{f\}} = \textrm{exp}(-\frac{1}{2}(\bar{x}_i^{d}-\bar{x}_i^{d^{\{f\}}})^\top \mathbf{P}(\bar{x}_i^{d}-\bar{x}_i^{d^{\{f\}}})),
\end{equation}
with $\mathbf{P}$ being a positive definite diagonal matrix representing the kernel widths. The terminal cost for the new task becomes \begin{equation}\label{thm2_final_cost_composition}
    h_i(\bar{x}_{i}^{t_f}) = -\log(\sum\nolimits_{f=1}^F \tilde{\omega}_i^{\{f\}} exp(-h_i^{\{f\}}(\bar{x}_{i}^{t_f}))),
\end{equation} where $\bar{x}^{t_f}$ denotes the boundary states, and $\tilde{\omega}_i^{\{f\}} = \frac{\bar{\omega}_i^{\{f\}}}{\sum_{f=1}^{F}\bar{\omega}_i^{\{f\}}} $ can be interpreted as the probability weights. The control law solving the new problem is obtained by composition of the existing controllers
$
    \bar{u}_i^{*}(\bar{x}_i, t) = \sum\nolimits_{f=1}^{F}\bar{W}_i^{\{f\}}(\bar{x}_i,t)\bar{u}_i^{*^{\{f\}}}(\bar{x}_i,t),
$
with $
    \bar{W}_i^{\{f\}}(\bar{x}_i,t) = \frac{\tilde{\omega}_i^{\{f\}}Z_i^{\{f\}}(\bar{x}_i,t)}{\sum\nolimits_{e=1}^{F} \tilde{\omega}_i^{\{e\}}Z_i^{\{e\}}(\bar{x}_i,t)}.
$
\end{theorem}
\begin{proof}
From composition of the final cost function in continuous-time given by equation \eqref{thm2_final_cost_composition}, and similar reasoning as in the discrete-time case, we have the compositionality of desirability functions holding  everywhere: \begin{equation}\label{eq:thm2_des_relation_interior}
    Z_{i}(\bar{x}_{i},t) = \sum\nolimits_{f = 1}^{F} \tilde{\omega}_i^{\{f\}}Z_{i}^{\{f\}}(\bar{x}_{i},t).
\end{equation} 

For the composite problem, the joint optimal control actions, introduced in equation \eqref{eq:joint_action_continuous}, can thus be reduced to
\begin{align*}
    &\quad \bar{u}_i^{*}(\bar{x}_i, t) = \bar{\sigma}_i \bar{\sigma}_i^\top \bar{B}_i^\top (\bar{x}_i) \frac{\nabla_{\bar{x}_i} Z_i(\bar{x}_i,t)}{Z_i(\bar{x}_i,t)}\\
    &= \bar{\sigma}_i \bar{\sigma}_i^\top \bar{B}_i^\top (\bar{x}_i) \frac{\nabla_{\bar{x}_i} \Big[{\sum\nolimits_{f=1}^{F} \tilde{\omega}_i^{\{f\}}Z_i^{\{f\}}(\bar{x}_i,t)}\Big]}{\sum\nolimits_{f=1}^{F} \tilde{\omega}_i^{\{f\}}Z_i^{\{f\}}(\bar{x}_i,t)}\\
    &= \frac{\sum\nolimits_{f=1}^{F}\bar{\sigma}_i \bar{\sigma}_i^\top \bar{B}_i^\top (\bar{x}_i)\nabla_{\bar{x}_i} \Big[\tilde{\omega}_i^{\{f\}}Z_i^{\{f\}}(\bar{x}_i,t)\Big]}{\sum\nolimits_{e=1}^{F} \tilde{\omega}_i^{\{e\}}Z_i^{\{e\}}(\bar{x}_i,t)}\\
    &= \frac{\sum\nolimits_{f=1}^{F}\bar{\sigma}_i \bar{\sigma}_i^\top \bar{B}_i^\top (\bar{x}_i)Z_i^{\{f\}}(\bar{x}_i,t)\nabla_{\bar{x}_i} \Big[\tilde{\omega}_i^{\{f\}}Z_i^{\{f\}}(\bar{x}_i,t)\Big]}{\sum\nolimits_{e=1}^{F} \tilde{\omega}_i^{\{e\}}Z_i^{\{e\}}(\bar{x}_i,t)Z_i^{\{f\}}(\bar{x}_i,t)}\\
    &= \sum\limits_{f=1}^{F}{\frac{\tilde{\omega}_i^{\{f\}}Z_i^{\{f\}}(\bar{x}_i,t)}{\sum\limits_{e=1}^{F} \tilde{\omega}_i^{\{e\}}Z_i^{\{e\}}(\bar{x}_i,t)}\bar{\sigma}_i \bar{\sigma}_i^\top \bar{B}_i^\top (\bar{x}_i)\frac{\nabla_{\bar{x}_i} Z_i^{\{f\}}(\bar{x}_i,t)}{Z_i^{\{f\}}(\bar{x}_i,t)}}\\
    &= \sum\nolimits_{f=1}^{F}\bar{W}_i^{\{f\}}(\bar{x}_i,t)\bar{u}_i^{*^{\{f\}}}(\bar{x}_i,t),
\end{align*}
with \begin{equation*}
    \bar{W}_i^{\{f\}}(\bar{x}_i,t) = \frac{\tilde{\omega}_i^{\{f\}}Z_i^{\{f\}}(\bar{x}_i,t)}{\sum\nolimits_{e=1}^{F} \tilde{\omega}_i^{\{e\}}Z_i^{\{e\}}(\bar{x}_i,t)}. 
\end{equation*}
\end{proof}
\section{SIMULATION RESULTS}\label{4_sec}

For demonstration of the proposed results, we validate the method on a cooperative UAV team governed by probability distribution in discrete-time as \eqref{joint_dynamics_discrete} and stochastic dynamics modeled by the It$\hat{\textrm{o}}$ diffusion process in continuous-time as \eqref{joint_dynamics_continuous}. In the team, agent 1 and agent 2 are closely coordinated by distance minimization using the running cost, while agent 3 is loosely connected only by the terminal cost. The networked architecture is illustrated in Figure \ref{fig_uav_team}.
\begin{figure}[h!]
\centering
\includegraphics[width=1.3in]{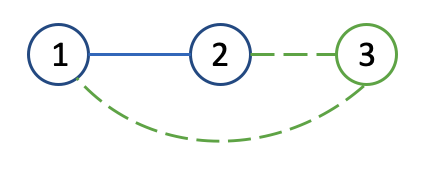}
\vspace*{-1mm}
\caption{Communication of the UAV team with agent 1 and agent 2 tightly-connected through running cost and agent 3 loosely-connected through the terminal cost.}
\label{fig_uav_team}
\end{figure}

\vspace*{-2mm}
\subsection{Discrete-time example of task generalization}
For discrete-time systems, we consider a cooperative UAV team described by the probability distribution model. The states are described by cells within a $5\times5$ grid, and the passive dynamics can be explained by the influence of random wind. We also consider the obstacle states (shaded region in Figures \ref{fig:comp_discrete}), which are penalized by larger state-related costs. The position of UAV $i$ is described by a state vector $x_i = [r_i, c_i]^\top$, where $r_i$ and $c_i \in \{1,2,3,4,5\}$. The goal of the cooperative UAV team is to achieve coordination between UAV 1 and UAV 2, by minimizing the distance between them for simulating the teamwork, and allow UAV 3 to fly to the target state individually under the controlled dynamics. Meanwhile, all the three UAVs are tasked with collision avoidance and cost minimization for optimal performance (i.e. travel time minimization).

For a factorization given by Figure \ref{fig_uav_team}, we have $
    \bar{x}_1 =[x_1;x_2]^\top,
    \bar{x}_2 = [x_1;x_2;x_3]^\top,
    \bar{x}_3 = [x_2;x_3]^\top.
$
Here, UAV 1 and UAV 2 are designed to fly cooperatively by getting close with each other, and the cost functions corresponding to the factorial subsystem states $\bar{x}_1$ and $\bar{x}_2$ are thus defined to contain terms penalizing on the row coordinate difference ($|r_2-r_1|$) and column coordinate difference ($|c_1-c_2|$).
The state-related running cost rates corresponding to three factorial subsystems are given by:
$
	q_1(\bar{x}_1) = 3.5 \cdot (|r_2 - r_1| + |c_2-c_1|) + o(x_1)\cdot o(x_2),
    q_2(\bar{x}_2) = 3.5 \cdot (|r_2 - r_1| 
    + |c_2-c_1|) 
     + o(x_1) \cdot o(x_2) \cdot o(x_3),
    q_3(\bar{x}_3) = 3.5 \cdot o(x_2) \cdot o(x_3), 
$
with $o(\cdot)$ denoting the state value (50 for the obstacle states and 2.5 for the other states).

We consider how the composite control action from existing controllers generalizes to a new task. In the first component problem, the target states of UAV 1, 2 and 3 are assigned to be $[2;2]^\top$, $[2;2]^\top$, and $[4;5]^\top$, respectively. In the second component problem, the target states of UAV 1, 2 and 3 are assigned to be $[3;3]^\top$, $[3;3]^\top$, and $[5;4]^\top$. We apply the iterative solver approach in \cite{todorov2009compositionality} for solving the desirability function. We choose the $P$ matrix in composite weight computation as a diagonal matrix with all identical elements representing the equal kernel width case. The composition weights are designed as in \eqref{eq:discrete_weight}. Then the composite control actions computed by weighing on the existing controllers according to Theorem \ref{theorem_dis}, solve a new optimal control problem, with target states of $[2;3]^\top$, $[2;3]^\top$, and $[5;5]^\top$ assigned to UAV 1, 2 and 3. 

The trajectories for the UAV team following the optimal control actions in the component and composite problem settings are provided in Figures \ref{fig:comp_discrete}, where `S' denotes
the initial states and `T'
denotes
the terminal states. From  Figure  \ref{fig:example_discrete_composite}, we notice that the composite control action solves the new problem successfully following some initial back and forth steps, while achieving the goal of cost minimization without re-computation of the optimal control actions. From the execution trajectory of composite control in Figure \ref{fig:example_discrete_composite}, UAV 1 and UAV 2 try to fly together most of time, while avoiding the obstacles as desired. Meanwhile, the solution is achieved in far less computational complexity than running the iterative solver algorithm again for the new task.



\begin{figure}[!hb]

\begin{minipage}{.5\linewidth}
\centering
\subfloat[Execution trajectory of the UAV team in component problem 1.]{\label{main:a}\includegraphics[scale=.35]{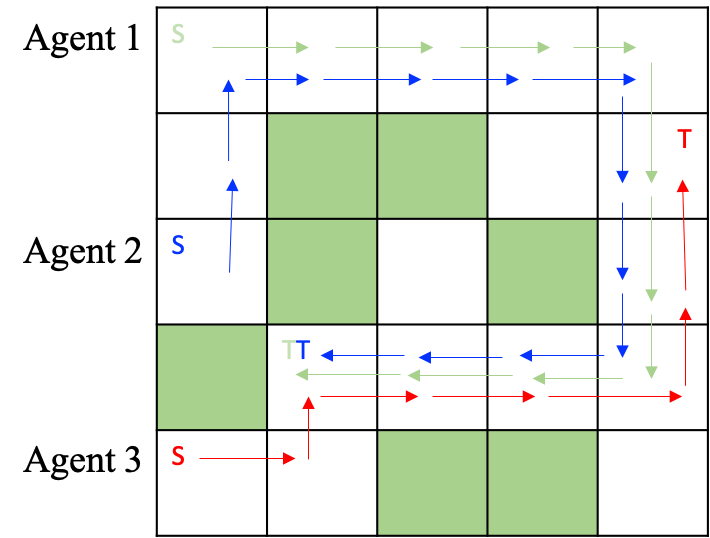}}
\end{minipage}%
\begin{minipage}{.5\linewidth}
\centering
\subfloat[Execution trajectory of the UAV team in component problem 2.]{\label{main:b}\includegraphics[scale=.35]{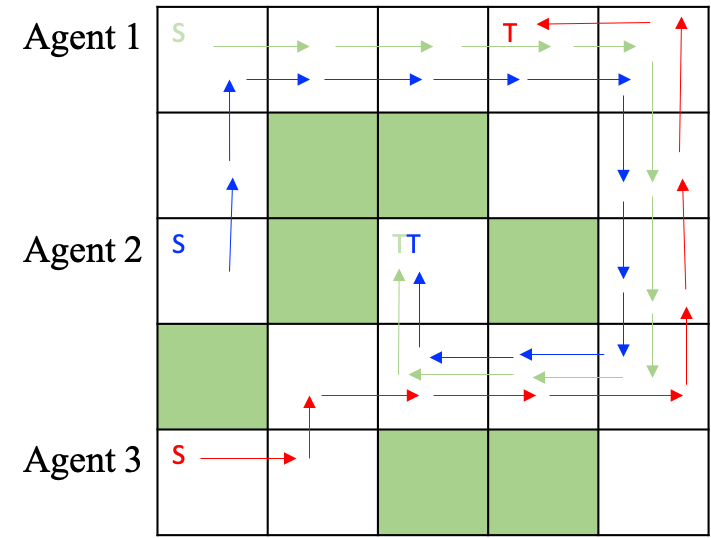}}
\end{minipage}\par\medskip
\centering
\vspace{-2mm}
\subfloat[Execution trajectory of the UAV team when the composite control law is applied in the new problem.]{\label{fig:example_discrete_composite}\includegraphics[scale=.4]{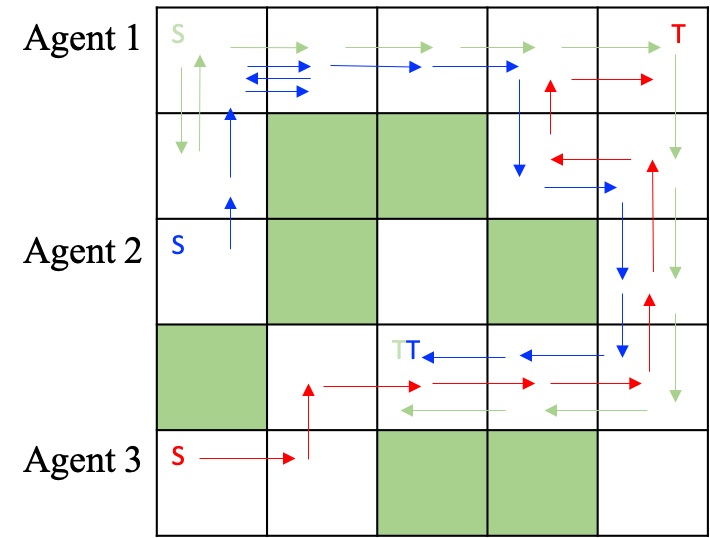}}

\caption{Demonstration of how component problem solutions generalize to a composite problem in discrete-time scenario.}
\label{fig:comp_discrete}
\end{figure}
\subsection{Continuous-time example of task generalization}
For continuous-time systems, we demonstrate our results in two different experimental settings. For both cases, we consider the UAV team with continuous-time dynamics described by the following equations as in \cite{wan2021distributed, yao2016hierarchical} :
	\begin{equation}
	\setlength\arraycolsep{1pt}
	{\scalefont{0.65} \left(\begin{matrix}
		dx_i\\
		dy_i\\
		dv_i\\
		d\varphi_i
		\end{matrix}\right) = 
		\left(\begin{matrix}
		v_i \cos \varphi_i\\
		v_i \sin \varphi_i\\
		0\\
		0
		\end{matrix}\right) dt +  \begin{pmatrix}
		0 & 0\\
		0 & 0\\
		1 & 0\\
		0 & 1
		\end{pmatrix} \left[ \left( \begin{matrix}
		u_i\\
		\omega_i
		\end{matrix}  \right) dt + \begin{pmatrix}
		\sigma_i & 0\\
		0 & \nu_i
		\end{pmatrix}dw_i
		\right], }
	\end{equation}
where $(x_i,y_i),v_i,\varphi_i$ denote the position coordinate, forward velocity and heading angle of the agent $i$, respectively, and $(x_i,y_i,v_i,\varphi_i)^\top$ is used as a state vector. Forward acceleration $u_i$ and angular velocity $\omega_i$ are the control inputs, and disturbance $w_i$ is a standard Brownian motion. We set the noise level parameters as $\sigma_i = 0.05$ and $v_i = 0.025$ throughout the simulation.

In both examples, the optimal solution to the component problems is achieved in a path-integral approximation framework. The $P$ matrix in the composite weight computation is chosen as an identity matrix representing the equal kernel width case. We compute the composite control law using the existing controllers according to Theorem \ref{theorem_cont}. 

\subsubsection{Example 1: All component problems only differ at the terminal cost}

In the first example, we consider the new problem holding the composition relationship on the final cost. We are interested in the final cost in the linear form of $
    h = \frac{d}{2}(|x-x_d|-c)+\alpha
$, and the three concerned component problems are designed with different sets of cost parameters $c, d$ and $\alpha$. The coordination between agents is considered in the running cost in the following form:
\begin{align}
    q_1(\bar{x}_1) &= 0.9\cdot (\|(x_1,y_1)-(x_1^{t_f},y_1^{t_f})\|_2-d_1^{\textrm{max}}) \\ 
    & \hspace{35pt}+ 1.5 \cdot(\|(x_1,y_1)-(x_2,y_2)\|_2-d_{12}^{\textrm{max}}),\nonumber\\
    q_2(\bar{x}_2)   &= 0.9\cdot(\|(x_2,y_2)-(x_2^{t_f},y_2^{t_f})\|_2-d_2^{\textrm{max}}) \\
    &\hspace{35pt} + 1.5\cdot(\|(x_2,y_2)-(x_1,y_1)\|_2-d_{21}^{\textrm{max}}),\nonumber\\
    q_3(\bar{x}_3) &= 0.9\cdot(\|(x_3,y_3)-(x_3^{t_f},y_3^{t_f})\|_2-d_3^{\textrm{max}}),
\end{align}
where $\|(x_i,y_i)-(x_i^{t_f},y_i^{t_f})\|_2$ calculates the distance to the goal position for UAV $i$, $\|(x_i,y_i)-(x_j,y_j)\|_2$ calculates the distance between UAV $i$ and UAV $j$,  $d_i^{\textrm{max}}$ denotes the distance between the initial position and target position for UAV $i$, and $d_{ij}^{\textrm{max}}$ denotes the initial distance between UAV $i$ and UAV $j$. These parameters can be tuned for improving  the stability and algorithm performance. 

For each component problem, one individual run is implemented towards given target states. For three UAVs with initial states $x_1^0 = (5,5,0.3,0)^\top$, $x_2^0 = (5,35,0.3,0)^\top$ and $x_3^0 = (5,20,0.3,0)^\top$, we want them to fly cooperatively towards the same target states of $(30,20,0,0)^\top$. The execution trajectories for the component problems are illustrated in Figure  \nolinebreak \ref{fig_cont_final_cost_differ_comp}, where `S' denotes
the initial states, and `T'
denotes
the terminal states. As the figure demonstrates, the trajectories are not obviously different from each other except noise, since the target states for different problems are identical, and the effect of terminal cost in obtaining the optimal control law is mild. The performance of executing the composite control actions on the new problem satisfying the final cost composition relationship in \eqref{eq:final_cost_relation} is given in Figure \ref{fig_cont_final_cost_differ_compos}.

\begin{figure}[!hb]
\centering
\includegraphics[width=2.4in]{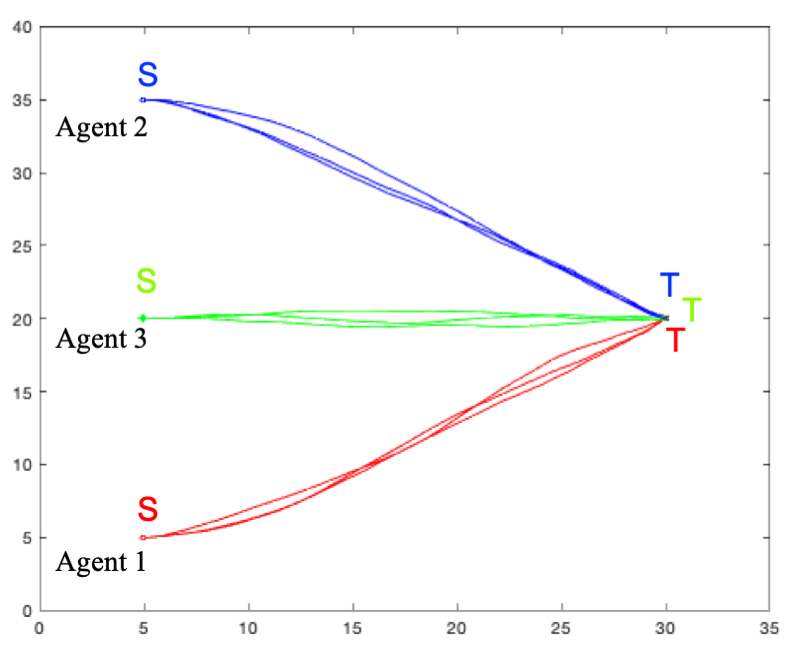}
\caption{Component problems for composition on the optimal control law when all problems differ only at the final cost (with lines in red, blue, and green denoting agents 1, 2, and 3, respectively).}
\label{fig_cont_final_cost_differ_comp}
\end{figure}
\begin{figure}[!hb]
\centering
\includegraphics[width=2.4in]{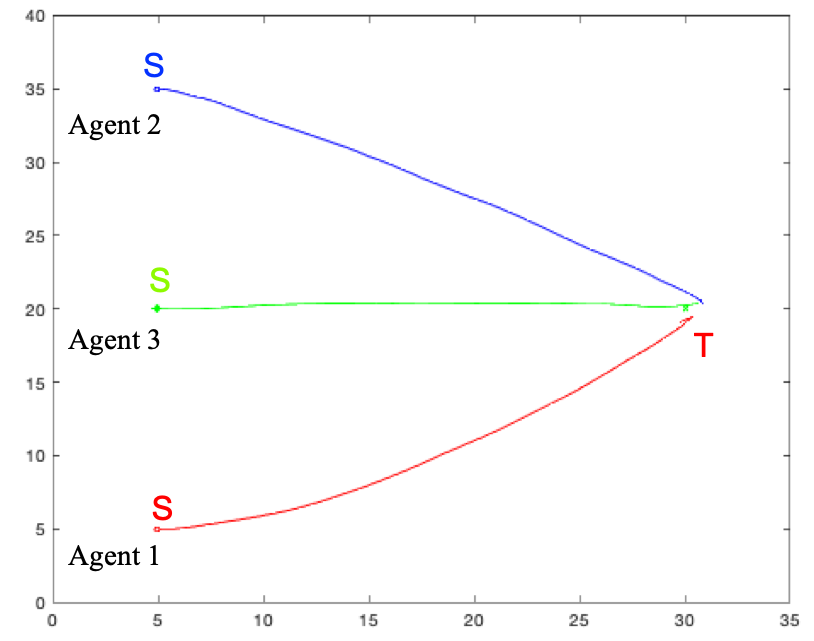}
\caption{Performance of applying the composite control actions directly computed from component control actions to the new problem with composite terminal cost.  }
\label{fig_cont_final_cost_differ_compos}
\end{figure}

\subsubsection{Example 2: All component problems differ both at the terminal cost and terminal states (generalization to new tasks)}
In this example, we consider the case when both terminal costs and terminal states differ among different component problems. The running costs and final costs follow the definition in Example 1, and the difference lies in the final state generalization. For the UAVs, the initial states are $x_1^0 = (10,10,0.3,0)^\top$, $x_2^0 = (10,30,0.3,0)^\top$ and $x_3^0 = (10,20,0.3,0)^\top$, and they are designed to fly cooperatively towards the same target: $x^{\{1\}} = (35,28,0,0)^\top$ and $x^{\{2\}} = (35,14,0,0)^\top$ in the two component problems; the trajectories are presented in Figure \ref{fig_cont_final_cost_final_state_differ_component}. The optimal control actions are obtained in the path integral framework introduced in \cite{theodorou2010generalized}.

\begin{figure}[h!]
\includegraphics[width=2.4in]{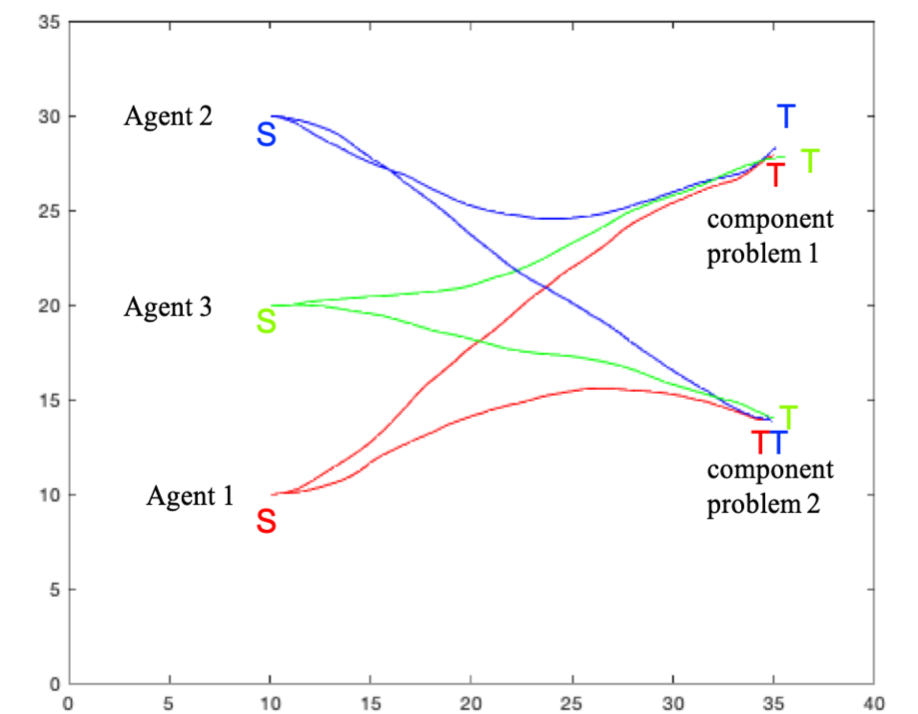}
\centering
\vspace*{-1mm}
\caption{Component problems targeted with different terminal states are considered for composition on the optimal control law (with lines in red, blue and green denoting agents 1, 2 and 3, respectively). These two component problems are set with target states of the upper and lower points, respectively. For each component problem, one individual run is implemented for composition purposes.}
\label{fig_cont_final_cost_final_state_differ_component}
\end{figure}

When the composite weights are calculated according to equation \eqref{eq:continuous_weight}, the composite control result can be generalized to a new task with target state of $x = (35,20,0,0)^\top$. The composite control actions associated with the new task are immediately achieved with existing controllers according to  Theorem \ref{theorem_cont}, and thus the approach provides a solution to a composite LSOC problem in a sample-efficient manner. The trajectory of the UAV 
team directly driven by the composite control law is shown in Figure \ref{fig_cont_final_cost_final_state_differ_composite}. 

\begin{figure}[h!]
\centering
\includegraphics[width=2.4in]{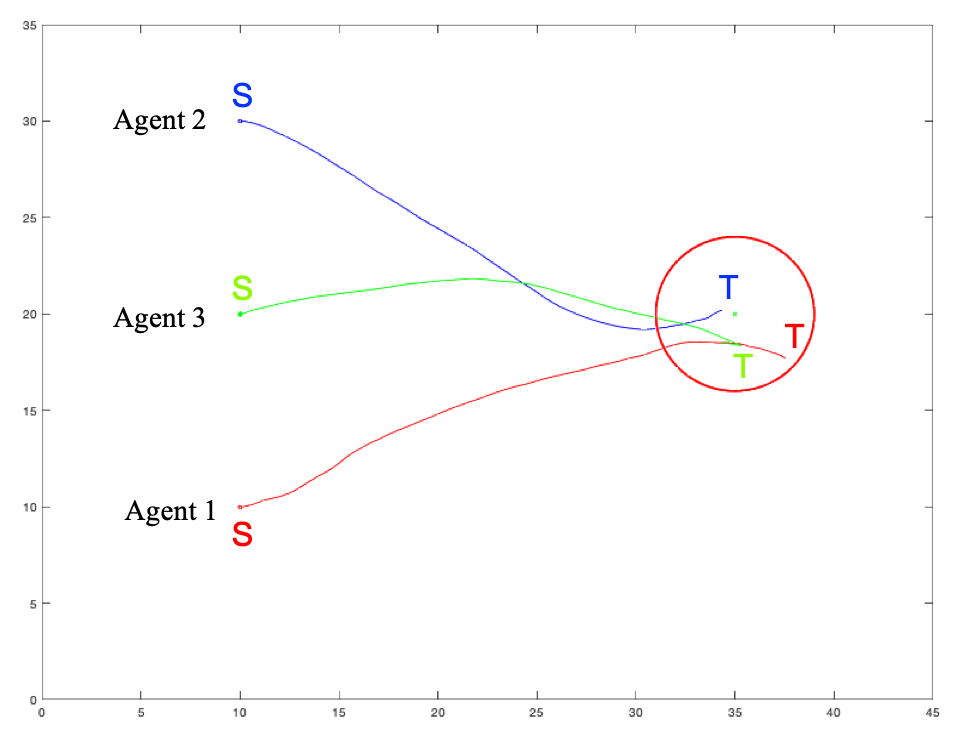}
\caption{Composite problem demonstrating the generalization on unlearned tasks (with lines in red, blue and green denoting agent 1, 2 and 3, respectively). The red circle centered at (35,20) (target coordinate for the UAV team) denotes the acceptable error region considering noise from the nature of stochasticity.}
\label{fig_cont_final_cost_final_state_differ_composite}
\end{figure}

In the discrete-time scenario, when the dimension of state space is finite, the exact composition on the value function is feasible, which enables the exact compositionality of optimal control actions. However, in the continuous-time  case, the state-space dimension is infinite. In this case   usually  the value function gets treated as a time-dependent function, and  the compositionality of optimal control law gets extended similarly to its counterpart in the discrete-time setting \cite{pan2015sample}. The disparity between the ending position and set targets in Figure \ref{fig_cont_final_cost_final_state_differ_composite} can be explained by the error introduced in applying this composite time-dependent control laws directly to the system in an offline manner. This issue in executing composite control laws directly may also be risky in 
high-dimensional single-agent problems except for fine-tuned parameters.  However, after parameter tuning and selecting optimal control actions independently for three factorial subsystems in multiple runs, the error can be restricted to a red circle shown in Figure \ref{fig_cont_final_cost_final_state_differ_composite} denoting the tolerance considering the presence of noise in the stochastic systems and the approximation error in  applied research. It is noteworthy that the result is achieved with far less computational complexity (away from resampling), which provides more attractive benefits as compared to the preciseness of control in some problem settings, when e.g. the component problem solution is obtained by sampling-based approach and is  computationally costly to achieve. 

\section{CONCLUSIONS AND FUTURE WORK}\label{5_sec}
Discrete-time and continuous-time distributed LSOC algorithms and composition theory on optimal control laws for networked MASs have been investigated in this paper. For both scenarios, existing controllers can be immediately composited to solve a new problem once certain weight conditions are satisfied on the final cost under the assumption of identical dynamics, cost rates, and set of interior states among all the component problems. 
The composite weights are computed via a square exponential kernel weight function measuring the target state differences,  and the composite control actions, achieved by weighing on the existing controllers, solve a new task in a sample-efficient manner.

There are some directions worthwhile for further investigation. For example, in the continuous-time task generalization setting, time-dependent control laws are proposed here, which makes it unrealistic to track control updates on time steps for composition, and this may lead to some  challenges as one faces when offline-computed control laws are directly applied to a new control problem. However, the error can be mitigated by using interpolation techniques on time. Meanwhile, for the computation of continuous-time optimal control laws, we applied the path-integral framework, and some approximation approaches such as random sampling estimator and Relative-Entropy Policy Search (REPS) algorithm are yet to be explored.

\section*{ACKNOWLEDGEMENTS}
This work is  supported by Air Force Office of Scientific Research (AFSOR), National Aeronautics and Space Administration (NASA) and National Science Foundation's National Robotics Initiative (NRI) and Cyber-Physical Systems (CPS) awards \#1830639, \#1932529 and \#1932288. 
\bibliography{reference}
\bibliographystyle{IEEEtran}

\addtolength{\textheight}{-12cm}   

\end{document}